\documentclass[a4paper,UKenglish]{lipics-v2019}
\usepackage{amsmath,amsthm,amsfonts,graphicx,hyperref,color,caption,subcaption,float,algpseudocode,soul,multirow,tikz,xspace,makecell,appendix}
\usepackage{algorithmicx}

\usepackage{microtype,amssymb}
\usepackage{framed}
\usepackage{enumitem}
\usepackage{thmtools,thm-restate}

\def\eps{\varepsilon}
\def\RR{{\mathbb R}}
\def\XX{{\mathbb X}}
\def\RRR{{\mathcal R}}

\def\MM{{\mathbb M}}
\def \d{{\mathrm d}}
\def \ddim{{\mathrm {ddim}}}
\newcommand{\nnm}{ball\xspace}
\newcommand{\nnms}{balls\xspace}
\newcommand{\ball}{B}
\newcommand{\stadium}{D}
\newcommand{\cylinder}{C}
\newcommand{\ccylinder}{R}

\newcommand{\predVEh}{P_1}
\newcommand{\predVEv}{P_2}
\newcommand{\predSLh}{P_3}
\newcommand{\predSLv}{P_4}
\newcommand{\predEs}{P_5}
\newcommand{\predEe}{P_6}
\newcommand{\predMh}{P_7}
\newcommand{\predMv}{P_8}

\newcommand{\etal}{\textit{e{}t~a{}l.}\xspace}

\DeclareMathOperator{\kde}{\textsc{kde}}
\title{The VC Dimension of Metric Balls under Fr\'echet and Hausdorff Distances}

\author{Anne Driemel}{University of Bonn, Germany}{driemel@cs.uni-bonn.de}{}{Anne Driemel thanks the Hausdorff Center for Mathematics for their generous support and the Netherlands Organization for Scientific Research (NWO) for support under Veni Grant 10019853.}

\author{André Nusser}{Max Planck Institute for Informatics \&\\ Saarbrücken Graduate School of Computer Science, Germany}{anusser@mpi-inf.mpg.de}{}{} 

\author{Jeff M. Phillips}{University of Utah, USA}{jeffp@cs.utah.edu}{}{Jeff Phillips thanks his support from NSF CCF-1350888, ACI-1443046, CNS-1514520, CNS-1564287, and IIS-1816149.  Part of the work was completed while visiting the Simons Institute for Theory of Computing.}

\author{Ioannis Psarros}{National \& Kapodistrian University of Athens, Greece}{ipsarros@di.uoa.gr}{}{This research is co-financed by Greece and the European Union (European Social Fund- ESF) through the Operational Programme $\ll$ Human Resources Development, Education and Lifelong Learning $\gg$ in the context of the project ``Strengthening Human Resources Research Potential via Doctorate Research'' (MIS-5000432), implemented by the State Scholarships Foundation (IKY).} 

\authorrunning{A.\ Driemel, A. Nusser, J.\ M.\ Phillips, I.\ Psarros}

\Copyright{Anne Driemel, Jeff M.\ Phillips, Ioannis Psarros}



\nolinenumbers 
\hideLIPIcs  

\EventEditors{Gill Barequet and Yusu Wang}
\EventNoEds{2}
\EventLongTitle{35th International Symposium on Computational Geometry (SoCG 2019)}
\EventShortTitle{SoCG 2019}
\EventAcronym{SoCG}
\EventYear{2019}
\EventDate{June 18--21, 2019}
\EventLocation{Portland, United~States}
\EventLogo{socg-logo}
\SeriesVolume{129}
\ArticleNo{0} 

\begin{document}

\maketitle

\begin{abstract}
The Vapnik-Chervonenkis dimension provides a notion of complexity for systems of sets. If the VC dimension is small, then knowing this can drastically simplify fundamental computational tasks such as classification, range counting, and density estimation through the use of sampling bounds. We analyze set systems where the ground set $X$ is a set of polygonal curves in $\mathbb{R}^d$ and the sets $\RRR$ are metric balls defined by curve similarity metrics, such as the Fr\'echet distance and the Hausdorff distance, as well as their discrete counterparts. We derive upper and lower bounds on the VC dimension that imply useful sampling bounds in the setting that the number of curves is large, but the complexity of the individual curves is small. Our upper bounds are either near-quadratic or near-linear in the complexity of the curves that define the ranges and they are  logarithmic in the complexity of the curves that define the ground set. 

\end{abstract}

\section{Introduction}

A \emph{range space} $(X,\RRR)$ (also called \emph{set system}) is defined by a ground set $X$ and a set of ranges $\RRR$, where each $r \in \RRR$ is a subset of $X$. 
A data structure for range searching answers queries for the subset of the input data that lies inside the query range. In range counting, we are interested only in the size of this subset. In our setting, a range is a metric ball defined by a curve and a radius. The ball contains all curves that lie within this radius from the center under a specific distance function (e.g., Fr\'echet or Hausdorff distance).

A crucial descriptor of any range space is its VC-dimension~\cite{VC71,She72,Sau72} and related shattering dimension, which we define formally below.  These notions quantify how complex a range space is, and have played fundamental roles in machine learning~\cite{Vap98,AB99}, data structures~\cite{CW89}, and geometry~\cite{H11,BG95}.  For instance, specific bounds on these complexity parameters are critical for tasks as diverse as neural networks~\cite{AB99,KM95}, art-gallery problems~\cite{Val98,GK14,LL17}, and kernel density estimation~\cite{JKPV11}.  

The last five years have seen a surge of interest into data structures for trajectory processing under the Fr\'echet distance, manifested in a series of publications
\cite{deBerg-2013fast, Gudmundsson-fqt-15, Berg-spq-15, AD18, giscup2017, BaldusB2017, DutschV17, BuchinDDM17, Driemel-lshc-17, astefanoaei2018multi, emiris2018}. Partially motivated by the increasing availability and quality of trajectory data from mobile phones, GPS sensors, RFID technology and video analysis~\cite{Laurila_MDC_2012, Zheng-rfid-2016, Gudmundsson-stats-2017}.
Initial results in this line of research, such as the approximate range counting data structure by de Berg, Gudmundsson and Cook~\cite{deBerg-2013fast}, use classical data structuring techniques. Afshani and Driemel extended their results and in addition showed lower bounds on the space-query-time trade-off in this setting~\cite{AD18}. In particular, they showed a lower bound which is exponential in the complexity of the curves for exact range searching. In 2017, ACM SIGSPATIAL, the premier conference for geographic information science, devoted their software challenge (GIS CUP) to the problem of range searching under the Fr\'echet distance~\cite{giscup2017}. Spurring further developments, the most recent results explore the use of heuristics~\cite{bringmannKN19} and randomization~\cite{ceccarelloD019}.

The Fr\'echet distance is a popular distance measure for curves. Intuitively, it can be defined using the metaphor of a person walking a dog, where the person follows one curve and the dog follows the other curve, and throughout their traversal they are connected by a leash of fixed length. The Fr\'echet distance corresponds to the length of the shortest dog leash that permits a traversal in this fashion. The Fr\'echet distance is very similar to the Hausdorff distance for sets, which is defined as the minimal maximum distance of a pair of points, one from each set, under all possible mappings between the two sets. The difference between the two distance measures is that the Fr\'echet distance requires the mapping to adhere to the ordering of the points along the curve. Both distance measures allow flexible associations between parts of the input elements which sets them apart from classical $L_p$ distances and makes them so suitable for trajectory data under varying speeds. 

Our contribution in this paper is a comprehensive analysis of the Vapnik-Chervonenkis dimension of the corresponding range spaces. The resulting VC dimension bounds, while being interesting in their own right, have a plethora of applications through the implied sampling bounds. We detail a range of implications of our bounds in Section~\ref{sec:implications}.

\section{Definitions}\label{prelim}
In this section, we formally define the distances between curves as well as VC-dimension and range spaces, so we can state our main results.  
This basic set up will be enough to prove our results for the discrete variants of the distance measures we consider. The basic proofs in the discrete setting also serve as a template for the proofs in the main part of the paper. Starting in Section \ref{prelim2} we provide more advanced geometric definitions and properties about VC dimension which we then use in our proofs on the continuous variants of the distance measures we consider.

\subsection{Distance measures} 
In the following, we define the Hausdorff distance, the discrete and the continuous Fr\'{e}chet distance, and the Weak Fr\'{e}chet distance. We denote by $\|\cdot\|$ the Euclidean norm $\|\cdot\|_2$.

\begin{definition}[Directed Hausdorff distance.] Let $X$, $Y$ be two  subsets of some metric space $(M,\d)$. The directed Hausdorff distance from $X$ to $Y$ is:
\[
\d_{\overrightarrow{H}} (X,Y ) = \sup_{u\in X} \inf_{v\in Y} \d(u, v).
\]
\end{definition}

\begin{definition}[Hausdorff distance.] Let $X$, $Y$ be two  subsets of some metric space $(M,\d)$. The Hausdorff distance between $X$ and $Y$ is:
\[
\d_{H} (X,Y ) = \max \{d_{\overrightarrow{H}}(X,Y),d_{\overrightarrow{H}}(Y,X) \} .
\]
\end{definition}

 \begin{definition}
 Given polygonal curves $V$ and $U$ with vertices $v_1, \ldots, v_{m_1}$ and $u_1, \ldots, u_{m_2}$ respectively, a traversal $T=(i_1,j_1),\ldots,(i_t,j_t)$ is a sequence of pairs of indices referring to a pairing of vertices from the two curves such that:
\begin{enumerate}
 \item $i_1,j_1=1$, $i_t=m_1$, $j_t=m_2$. 
 \item $\forall (i_k, j_k)\in T:$ $i_{k+1}-i_k \in \{0,1\}$ and $j_{k+1}-j_k \in \{0,1\}$.
 \item $\forall (i_k, j_k)\in T:$ $(i_{k+1}-i_k)+(j_{k+1}-j_k)\geq1$.
\end{enumerate} 
 \end{definition} 

\begin{definition}[Discrete Fr\'{e}chet distance]\label{Ddist}
 Given polygonal curves $V$ and $U$ with vertices $v_1, \ldots, v_{m_1}$ and $u_1, \ldots, u_{m_2}$ respectively, we define the Discrete Fr\'{e}chet Distance between $V$ and $U$ as the following function:
 \[
 \d_{dF}(V,U)= \min_{T\in\mathcal{T}} ~ \max_{(i_k,j_k)\in T} \| v_{i_k}-u_{j_k}\|  ,
 \]
 where $\mathcal{T}$ denotes the set of all possible traversals for $V$ and $U$. 
 \end{definition}
 
Any polygonal curve $V$ with vertices $v_1,\ldots,v_{m_1}$ and edges $\overline{v_1 v_2},\ldots, \overline{v_{m_1-1}v_{m_1}}$  has a 
uniform parametrization that allows us to view it as a parametrized curve $v:~ [0,1]\mapsto \RR^2$.

\begin{definition}[Fr\'{e}chet distance]
Given two parametrized curves $u,v:~ [0,1]\mapsto \RR^2 $, their Fr\'{e}chet distance is defined as follows:
\[
\d_F(u,v)=\min_{f: [0,1] \mapsto [0,1] \atop g:[0,1] \mapsto [0,1]} ~\max_{\alpha\in[0,1]} \|v(f(\alpha))- u (g(\alpha))\|,
\]
where $f$ and $g$ range over all continuous, non-decreasing functions with $f(0) = g(0) = 0$, and $f(1) = g(1) = 1$.
\end{definition}

\begin{definition}[Weak Fr\'{e}chet distance]
Given two parametrized curves $u,v:~ [0,1]\mapsto \RR^2 $, their Weak Fr\'{e}chet distance is defined as follows:
\[
\d_{wF}(u,v)=\min_{f: [0,1] \mapsto [0,1] \atop g: [0,1] \mapsto [0,1]} ~\max_{\alpha\in[0,1]} \|v(f(\alpha))- u (g(\alpha))\|,
\]
where $f$ and $g$ range over all continuous functions with $f(0) = g(0) = 0$, and $f(1) = g(1) = 1$.
\end{definition}


\subsection{Range spaces} 
Each range space can be defined as a pair of sets $(X,\RRR)$, where $X$ is the \textit{ground set} and $\RRR$ is the \textit{range set}. 
Let $(X,\RRR)$ be a range space. For $Y\subseteq X$, we denote: 
\[
\RRR_{|Y}= \{ R \cap Y \mid R \in \RRR \}.
\]
If $\RRR_{|Y} $ contains all subsets of $Y$, then $Y$ is \textit{shattered} by $\RRR$. 
\begin{definition}[Vapnik-Chernovenkis dimension]
The Vapnik-Chernovenkis dimension~\cite{Sau72, She72, VC71} (VC dimension) of $(X,\RRR)$ is the maximum cardinality of a shattered subset of $X$.
\end{definition}

\begin{definition}[Shattering dimension]
The shattering dimension of $(X,\RRR)$ is the smallest $\delta$ such that, for all m, 
\[ \max_{\substack{B\subset X\\|B|=m}} | \RRR_{|B} | =O(m^\delta).\]
\end{definition}
It is well-known that for a range space $(X,\RRR)$ with VC-dimension $\nu$ and shattering dimension $\delta$ that $\nu \leq O(\delta \log \delta)$ and $\delta = O(\nu)$.  So bounding the shattering dimension and bounding the VC-dimension are asymptotically equivalent within a log factor.  For a proof of this and other basic facts on range spaces we refer the reader to the textbook of Har-Peled~\cite{H11}.


\begin{definition}[Dual range space]
Given a range space $(X,\RRR)$, for any $p\in X$, we define 
\[
\RRR_p=\{R \mid R \in \RRR, p \in R \}. 
\]
The dual range space of $(X,\RRR)$ is the range space $\left( \RRR,\{\RRR_p\mid p \in X \} \right) $.
\end{definition}
It is a well-known fact that if a range space has VC dimension $\nu$, then the dual range space has VC dimension $\leq 2^{\nu+1}$ (see e.g.~\cite{H11}).

There are many techniques for bounding the VC dimension of geometric range spaces.  For instance when the ground set is $\RR^d$ and the ranges are defined by inclusion in halfspaces, then the range space and its dual range space are isomorphic and both have VC-dimension and shattering dimension $d$.  When the ranges are defined by inclusion in balls, then the VC-dimension and shattering dimension is $d+1$, and the dual range spaces have bounds of $d$~\cite{H11}.    
It is also for instance known~\cite{BEHW89} that the composition ranges formed as the $k$-fold union or intersection of ranges from a range space with bounded VC-dimension $\nu$ induces a range space with VC-dimension $O(\nu k \log k)$, and it was recently shown by Csik{\'o}s et al.\ that this is tight for even some simple range spaces such as those defined by halfspaces~\cite{CKM18,csikos2019tight}.  
More such results are deferred to Section \ref{prelim2}.  

%
%
%
%

\subsection{Range spaces induced by distance measures}

Let $(M,\d)$ be a pseudometric space. We define the \textit{\nnm} of radius $r$ and center $p$, under the distance measure $\d$, as the following set:
\[
b_{\d}(p,r)=\{ x\in M \mid \d(x,p) \leq r \},
\]
where $p\in M$.
The doubling dimension of a metric space $(M,\d)$, denoted as $\ddim(M,\d)$, is the smallest integer $t$ such that any ball can be covered by at most $2^t$ balls of half the radius.

In this paper, we study the VC dimension of variants of range spaces $(X,\RRR)$ induced by pseudometric spaces\footnote{While we may use the term \emph{metric} or \emph{pseudometric} to define the range, our methods do not assume any metric properties of the inducing distance measure.} $(M,\d)$ by setting $X=M$ and \[\RRR = \{b_{\d}(p,r) ~|~ r \in \RR, r > 0, p \in M \}.\] 
It is a reasonable question to ask whether the doubling dimension of a metric space influences the VC dimension of the induced range space. In general, a bounded doubling dimension does not imply a bounded VC dimension of the induced range space and vice versa. 
Recently, Huang et al.~\cite{huang2018epsilon} showed that if we allow a small $(1+\eps)$-distortion of the distance function $\d$, the shattering dimension can be upper bounded by $O(\eps^{-O(\ddim(M,\d))})$. It is conceivable that the doubling dimension of the metric space of the Discrete Fr\'echet distance and Hausdorff distance is bounded, as long as the underlying metric has bounded doubling dimension. However, for the continuous Fr\'echet distance, the doubling dimension is known to be unbounded~\cite{DriemelKS16}. Moreover, we will see that much better bounds can be obtained by a careful study of the specific distance measure. 

Specifically, we study an \emph{unbalanced} version of the above range space, in the sense that we distinguish between the complexity of objects of the ground set and the complexity of objects defining the ranges. To this end, we define, for any integers $d$ and $m$, $\XX_m^d := \left(\RR^d\right)^m$ and we treat the elements of this set as ordered sets of points in $\RR^d$ of size $m$. Formally, we study range spaces with ground set $\XX_m^d$ and a range set of the form
 \[\RRR_{\d,k} = \left\{b_{\d}(p,r) \cap \XX^d_{m} ~|~ r \in \RR, r > 0, p \in \XX^d_k \right\}\] 
 under different variants of the Fr\'echet and the Hausdorff distance. We emphasize that the range space consists of ranges of all radii.

\section{Our Results} 
Table~\ref{tab:results} shows an overview of our bounds.
For metric balls defined on point sets (resp. point sequences) in $\mathbb{R}^d$ we show that the VC dimension is at most near-linear in $dk$, the complexity of the ball centers that define the ranges, and at most logarithmic in $dm$, the complexity of point sets of the ground set.
Our lower bounds show that these bounds are almost tight in all parameters $k$, $d$, and $m$.  
For the Hausdorff distance, where the ground set $X$ consists of continuous polygonal curves in $\mathbb{R}^d$, we show an upper bound that is quadratic in $k$, quadratic in $d$ and logarithmic in $m$.  The same bound holds for the Fr\'echet distance, where the ground set consists of sets of line segments in $\mathbb{R}^d$.  We obtain slightly better bounds in $k$ for the Weak Fr\'echet distance.
Our lower bounds extend to the continuous case, but are only tight in the dependence on $m$ -- the complexity of the ground set.

\begin{table}[ht]
    \centering\def\arraystretch{1.3}
    \begin{tabular}{|c|c|c|c|}
    \hline
        \multirow{2}{*}{discrete} & 
        \multirow{1}{*}{Hausdorff} & 
        \multirow{2}{*}{ $O(d k\log(d k m))$ (Theorems~\ref{ThmHaus},\ref{ThmDFD}) }
          & \multirow{5}{3.6cm}{\centering $\Omega(\max(d k \log k, \log d m))$\newline ($d \geq 4$, Theorem~\ref{ThmLbsExt})\\ \vspace{0.5\baselineskip}
           $\Omega(\max(k, \log m))$\newline ($ d\geq 2$, Theorem~\ref{ThmLbs})
          } \\
        \cline{2-2}
        & \multirow{1}{*}{Fr\'echet} & & \\
        \cline{1-3}
        \multirow{3}{1.5cm}{\centering continuous} &
        \multirow{1}{*}{ Hausdorff }& $O(d^2 k^2\log(dkm))$ (Theorem \ref{Thaus3}) & \\
        \cline{2-3} & weak Fr\'echet & $O(d^2 k\log(dkm))$ (Theorem \ref{ThmwF}) & \\
        \cline{2-3}
        & Fr\'echet & $O(d^2 k^2\log(dkm))$ (Theorem~\ref{TFr}) & \\
        \hline
    \end{tabular}
    \caption{Our bounds on the VC dimension of range spaces of the form  $(\XX^d_m,\RRR{}_{\d,k})$, for $\d$ being the distance measures in the table. In the first column we distinguish between $\XX^d_m$ consisting of \emph{discrete} point sequences vs.\ $\XX^d_m$  consisting of \emph{continuous} polygonal curves. The lower bounds hold for all distance measures in this table.}
    \label{tab:results}
\end{table}

While the VC dimension bounds for the discrete Hausdorff and Fr\'echet metric balls may seem like an easy implication of composition theorems for VC dimension~\cite{BEHW89,CKM18}, we still find three things about these results remarkable: 
\begin{enumerate}
\item First, for Fr\'echet variants, there are $\Theta(2^k 2^m)$ valid alignment paths in the free space diagram.  And one may expect that these may materialize in the size of the composition theorem.  Yet by a simple analysis of the shattering dimension, we show that they do not.  
\item Second, the VC dimension only has logarithmic dependence on the size $m$ of the curves in the ground set, rather than a polynomial dependence one would hope to obtain by simple application of composition theorems.  This difference has important implications in analyzing real data sets where we can query with simple curves (small $k$), but may not have a small bound on the size of the curves in the data set (large $m$).  
\item Third, for the continuous variants, the range spaces can indeed be decomposed into problems with ground sets defined on line segments.  However, we do not know of a general $d$-dimensional bound on the VC-dimension of range space with a ground set of segments, and ranges defined by segments within a radius $r$ of another segment.  We are able to circumvent this challenge with circuit-based methods to bound the VC-dimension and careful predicate design.
\end{enumerate}

\section{Our Approach}

Our methods use the fact that both the Fr\'echet distance and the Hausdorff distance are determined by one of a discrete set of events, where each event involves a constant number of simple geometric objects. For example, it is well known that the Hausdorff distance between two discrete sets of points is equal to the distance between two points from the two sets. The corresponding event happens as we consider a value $\delta>0$ increasing from $0$ and we record which points of one set are contained in which balls of radius $\delta$ centered at points from the other set.
The same phenomenon is true for the discrete Fr\'echet distance between two point sequences. In particular, the so-called free-space matrix which can be used to decide whether the discrete Fr\'echet distance is smaller than a given value $\delta$ encodes exactly the information about which pairs of points have distance at most $\delta$. 
The basic phenomenon remains true for the continuous versions of the two distance measures if we extend the set of simple geometric objects to include line segments and if we also consider triple intersections. 
Each type of event can be translated into a range space of which we can analyze the VC dimension. Together, the product of the range spaces encodes the information, which curves lie inside which metric balls, in the form of a set system. This representation allows us to prove bounds on the VC dimension of metric balls under these distance measures.

\section{Basic Idea: Discrete Fr\'{e}chet and Hausdorff}

In this section we prove our upper bounds in the discrete setting.
Let $\XX_m^d = \left(\RR^d\right)^m$; we treat the elements of this set as ordered sets of points in $\RR^d$ of size $m$. The range spaces that we consider in this section are defined over the ground set $\XX_m^d$ and the range set of \nnms{} under either the Hausdorff or the Discrete Fr\'{e}chet distance.  
The proofs in the proceeding sections all follow the basic idea of the proof in the discrete setting.

\begin{theorem} \label{ThmHaus}
Let 
$(\XX^d_m,\RRR_{dH,k})$ be the range space with $\RRR_{dH,k}$ being the set of all \nnms{} under the Hausdorff distance centered at point sets in $\XX_k^d$. The VC dimension is $O\left(dk \log(dk m)\right)$.
\end{theorem}
\begin{proof}
Let $\{S_1,\ldots,S_t\} \subseteq \XX^d_m$ and $S=\bigcup_i S_i$; we define $S$ so that it ignores the ordering with each $S_i$ and is a single set of size $tm$.  
Any intersection of a Hausdorff ball with $\{S_1,\ldots,S_t\}$ is uniquely defined by a set $\{ B_1\cap S,\ldots,B_k \cap S \}$, where $B_1,\ldots,B_k$ are balls in $\RR^d$. To see that, notice that the discrete Hausdorff distance between two sets of points is uniquely defined by the distances between points of the two sets. 

Consider the range space $(\RR^d, \mathcal{B})$, where $\mathcal{B}$ is the set of balls in $\RR^d$. We know that the shattering dimension is $d+1$ \cite{H11}. Hence,
\[
\max_{S\subseteq \RR^d,|S|=tm} |\mathcal{B}_{|S}| = O((tm)^{d+1}).
\]
This implies that
$
|\{  \{B_1\cap S,\ldots,B_k \cap S\}  \mid B_1,\ldots,B_k \text{ are balls in }\RR^d \}| \leq  O((tm)^{(d+1)k}),$
and hence\footnote{for $u>\sqrt{e}$ if $x/\ln(x) \leq u$ then $x\leq 2u \ln u$. Hence, if $tm /\log(tm) \leq dkm $, then $tm =O( dkm \log (dkm))$.},
\[
2^t \leq  O((tm)^{(d+1)k})
\implies t = {O}\left(dk \log (dkm)\right). \qedhere
\]
\end{proof}

\begin{theorem}\label{ThmDFD}
Let 
$(\XX_m^d,\RRR_{dF,k})$ be the range space with  $\RRR_{dF,k}$ being the set of all \nnms under the Discrete Fr\'{e}chet distance centered at polygonal curves in $\XX_k^d$. The VC dimension is $O\left(dk \log(dk m)\right)$.
\end{theorem}
\begin{proof}
Let $\{S_1,\ldots,S_t\} \subseteq \XX_m^d$ and $S=\bigcup_i S_i$.  
Any intersection of a Discrete Fr\'{e}chet ball with $\{S_1,\ldots,S_t\}$ is uniquely defined by a sequence $ B_1\cap S,\ldots,B_k \cap S $, where $B_1,\ldots,B_k$ are balls in $\RR^d$. The number of such sequences can be bounded by $ O((tm)^{(d+1)k})$ as in the proof of Theorem~\ref{ThmHaus}.
Enforcing that a sequence contains a valid alignment path only reduces the number of possible distinct sets formed by $t$ curves, and it can be determined using these intersections and the two orderings of $B_1, \ldots, B_k$ and of vertices within some $S_j \in \XX_m^d$.  
\end{proof}


\section{Preliminaries}
\label{prelim2}

In this section, we provide a more advanced set of geometric primitives and other technical known results about the VC-dimension.  We also derive some simple corollaries.  Additionally, we provide some basic results about the distances which will couple with the geometric primitives in our proofs for continuous distance measures.  

We again consider a ground set $\XX_m^d = \left(\RR^d\right)^m$ which we treat as a set of  polygonal curves with points in $\RR^d$ of size $m$.  Given such a curve $s \in \XX_m^d$, let $V(s)$ be its ordered set of vertices and $E(s)$ its ordered set of edges.

\subsection{A simple model of computation}
We consider a model of computation that will be useful for modeling primitive geometric sets, and in turn bounding the VC-dimension of an associated range space.  These will be useful in that they allow the invocation of powerful and general tools to describe range spaces defined by distances between curves.  We allow the following operations, which we call \emph{simple operations}:
\begin{itemize}
    \item the arithmetic operations $+,-,\times,$ and $/$ on real numbers,
    \item jumps conditioned on $>, \geq,< , \leq ,=,$ and $\neq$ comparisons of real numbers, and 
    \item output $0$ or $1$.
\end{itemize}
We say a function requires $t$ simple operations if it can be computed with a circuit of depth $t$ composed only of these simple operations.  Notably, the lack of a square-root operator creates some challenges when dealing with non-linear geometric objects. Therefore, we prove the following technical lemma showing that we can compare certain expressions involving square roots without computing them explicitly, i.e., only simple operations are needed for the comparison.

\begin{lemma}\label{lem:comparison}
Consider values $a,b,c,d \in \RR$ with $b,d \geq 0$. We can compute the truth value of $a + \sqrt{b} \leq c + \sqrt{d}$ and $a + \sqrt{b} \geq c + \sqrt{d}$ using $O(1)$ simple operations.
\end{lemma}
\begin{proof}
It suffices to prove the case of $a + \sqrt{b} \leq c + \sqrt{d}$, as $a + \sqrt{b} \geq c + \sqrt{d}$ is analogous. We simply show that this comparison is equivalent to a comparison involving only a constant number of simple operations starting from the values $a, b, c, d$. If $a = c$, then $\sqrt{b} \leq \sqrt{d}$ is equivalent to $b \leq d$ and we are done. Assuming $a < c$, we get
\begin{align*}
a + \sqrt{b} \leq c + \sqrt{d} &\iff \sqrt{b} \leq (c-a) + \sqrt{d}\\
&\iff b \leq (c-a)^2 + 2(c-a)\sqrt{d} + d\\
&\iff b - (c-a)^2 - d \leq 2(c-a)\sqrt{d}
\end{align*}
The second equivalence holds because both sides are at least $0$. Now, note that the right side of the last inequality is at least $0$ and thus, if the left side is negative (which we can check using $O(1)$ simple operations), we are done. Thus, assume the left side is at least $0$. Then we can square both sides and obtain a comparison involving only simple operations.
Now, if $c < a$, we can do an analogous calculation, where we subtract $c$ instead of $a$ in the first equivalence. As testing $c < a$ is a simple operation, we can determine which case we are in.
\end{proof}

\subsection{Geometric primitives} 
For any $p\in \RR^d$ we denote by $\ball_r(p)$ the ball of radius $r$, centered at $p$. 
For any two points $s,t \in \RR^d$, we denote by $\overline{st}$ the line segment from $s$ to $t$. Whenever we store such a line segment, for technicalities within the lemma below, we store the coordinates of its endpoints $s$ and $t$.  
For any two points $s,t \in \RR^d$, we define the stadium centered at $\overline{st}$,  $\stadium_r(\overline{st})=\left\{x\in \RR^d \mid \exists p \in \overline{st}, ~\|p-x\|\leq r \right\}$. 
For any two points $s,t \in \RR^d$, we define a cylinder
$\cylinder_{r}(\overline{st})=\left\{x\in \RR^d \mid \exists p \in \ell(\overline{st}),~ \|p-x\|\leq r  \right\}$, where $\ell(\overline{st})$ denotes the line supporting the edge $\overline{st}$. 
Finally, for any two points $s,t \in \RR^d$, we define the capped cylinder centered at $\overline{st}$: 
$\ccylinder_r(\overline{st})=\{p+u \mid p \in \overline{st} \text{ and } u \in \RR^d \text{ s.t. } \|u\|\leq r \text{ and } \langle t-s,u\rangle =0\}$.  

\begin{figure}
\label{fig:shapes}
\includegraphics[width=\linewidth]{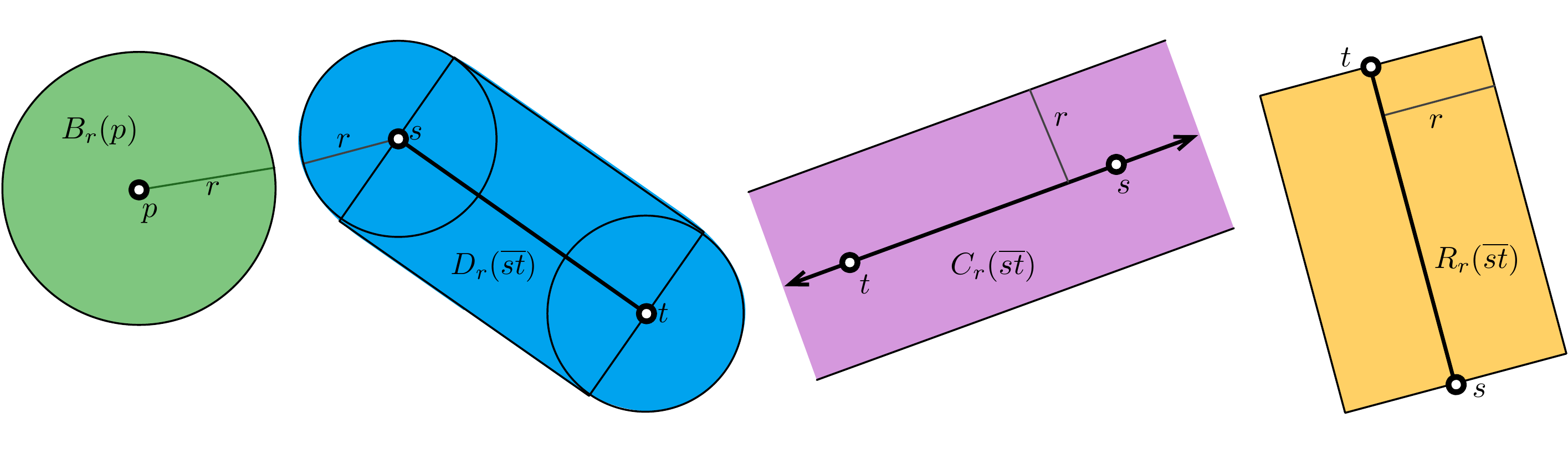}
\caption{Illustration of basic shapes in $\RR^2$, from left to right: a ball $\ball_r(p)$, a stadium $\stadium_r(\overline{st})$, a cylinder $\cylinder_r(\overline{st})$, and a capped cylinder $\ccylinder_r(\overline{st})$. }
\end{figure}

For each of these geometric sets, we can determine if a point $x \in \RR^d$ is in the set with a constant number of operations under a simple model of computation.  

\begin{lemma}
\label{lem:shapes}
For a point $x \in \RR^d$, and any set of the form $\ball_r(p)$, $\stadium_r(\overline{st})$, $\cylinder_r(\overline{st})$, or $\ccylinder_r(\overline{st})$, we can determine if $x$ is in that set (returns $1$, otherwise $0$) using $O(d)$ simple operations.  
\end{lemma}
\begin{proof}
For ball $\ball_r(p)$ we can compute a distance $\|x-p\|^2$ in $O(d)$ time, and determine inclusion with a comparison to $r^2$.  
For cylinder $\cylinder_r(\overline{st})$ we can compute the closest point to $x$ on this line as 
\[\pi_{\overline{st}}(x) = t + \frac{(s-t) \langle (s-t), x \rangle}{\|s-t\|^2} .\]  
Then we can determine inclusion by comparing $\|\pi_{\overline{st}}(x)-x\|^2$ to $r^2$.  
For capped cylinder $\ccylinder_r(\overline{st})$ we also need to compare $\|\pi_{\overline{st}}(x)  - t\|^2$ and $\|\pi_{\overline{st}}(x) -s\|^2$ to see if either of these terms is greater than $\|s-t\|^2$.
For stadium $\stadium_r(\overline{st})$ we determine inclusion if any $x$ is in any of $\ccylinder_r(\overline{st})$, $\ball_r(s)$ or $\ball_r(t)$.  
\end{proof}

\subsection{Bounding the VC-Dimension} 
For range spaces defined on continuous curves, our proofs use a powerful theorem from Goldberg and Jerrum~\cite{GJ95} as improved and restated by Anthony and Bartlett~\cite{AB99}.  It allows one to easily bound the VC-dimension of geometric range spaces under our simple model of computation.  

\begin{theorem}[Theorem 8.4 \cite{AB99}]
\label{theorem:predicatevc}
Suppose $h$ is a function from $\RR^d \times \RR^n$ to $\{0,1\}$ and let 
\[ H=\{x \mapsto h(\alpha,x):~\alpha \in \RR^d\}
\]
be the class determined by $h$. Suppose that $h$ can be computed by an algorithm that takes as input the pair $(\alpha,x) \in \RR^d \times \RR^n$ and returns $h(\alpha,x)$ after no more than $t$ simple operations.  
Then, the VC dimension of $H$ is $\leq 4d(t+2)$.
\end{theorem}

An example implication can be seen for geometric sets via Lemma \ref{lem:shapes}.  Note that this implies any VC dimension upper bound proved in this approach applies to both the range space and its dual range space because the function $h$ is unchanged and the ranges can still be described by $O(d)$ real coordinates.

\begin{corollary}
\label{cor:shapes-VC}
For range spaces defined on $\RR^d$ with geometric sets $\ball_r(p)$, $\stadium_r(\overline{st})$, $\cylinder_r(\overline{st})$, or $\ccylinder_r(\overline{st})$ as ranges, the VC dimension is $O(d^2)$.  
The same $O(d^2)$ VC dimension bound holds for the corresponding dual range spaces, with ground sets as the geometric sets, and ranges defined by stabbing using points in $\RR^d$.  
\end{corollary}

Note that these bounds are not always tight.  Specifically, because the VC-dimension for ranges defined geometrically by balls $\ball_r(p)$ is $O(d)$~\cite{H11}.  Moreover, the VC-dimension of range spaces defined by cylinders $\cylinder_r(\overline{st})$ is known to be $O(d)$~\cite{AIKU10}.  The ranges defined by capped cylinders $\ccylinder_r(\overline{st})$ are the intersection of a cylinder and two halfspaces, each with VC-dimension $O(d)$ and hence by the composition theorem~\cite{BEHW89}, this full range spaces also has VC-dimension $O(d)$.  Finally, the stadium $\stadium_r(\overline{st})$ is defined by the union of a capped cylinder $\ccylinder_r(\overline{st})$ and two balls $\ball_r(s)$ and $\ball_r(t)$; hence again by the composition theorem~\cite{BEHW89}, its VC-dimension is $O(d)$.  

However, it is not clear that these improved bounds hold for the dual range spaces, aside for the case of $\ball_r$.  
Moreover, when the ground set $X$ of the range space $(X, \RRR)$ is not $\RR^d$, then we need to be cautious in using the $k$-fold composition theorem~\cite{BEHW89}, which bounds the VC-dimension of complex range spaces derived as the logical intersection or union of simpler range spaces with bounded VC-dimension.  In the case of a ground set $X=\RR^d$, logical and geometric intersections are the same, but for other ground sets (like dual objects, or line segments $\XX^d_2$) this is not necessarily the case.  For instance, a line segment $e \in \XX^d_2$ may intersect a ball $\ball_r$ and also a halfspace $H$ while not intersecting the intersection $\ball_r \cap H$.  


\subsection{Representation by predicates}
\label{predicates}

In order to prove bounds on the VC dimension of range spaces defined on continuous curves, we establish sets of geometric predicates which are sufficient to determine if two curves have distance at most $r$ to each other.  Analyzing the range spaces associated with these predicates (over all possible radii $r$) allows us to compose them further and to establish VC dimension bounds for the range space induced by the corresponding distance measure.  
For the Fr\'{e}chet and Weak Fr\'{e}chet distance, the predicates mirror those used in range searching data structures~\cite{AD18,arxAD17}.  And for the Hausdorff distance on continuous curves, the predicates are derived from the Voronoi diagram~\cite{Alt-HausdorffDistance-1995}.  The technical challenges for each case are similar, but require different analyses.

\section{The Hausdorff distance}
\label{section:haus}
We consider the range space $(\XX^d_m, \RRR_{H_k}^r)$, where $\RRR_{H_k}^r$ denotes the set of all \nnms, of radius $r$ centered at curves in $\XX^d_k$, under the symmetric Hausdorff distance.\footnote{The proofs in this section are written for polygonal curves in $\XX^d_{m}$, but they readily extend to (not-necessarily connected) sets of line segments in $\RR^d$ of cardinality $m'=\frac{m-1}{2}$.}  
We also consider the same problems under both directed versions of the Hausdorff distance, and their induced range spaces 
$(\XX^d_m, \RRR_{\overrightarrow{H}_k}^r)$ and $(\XX^d_m, \RRR_{\overleftarrow{H}_k}^r)$.  

\subsection{Hausdorff distance predicates}
\label{predicates-hausdorff}

According to Alt, Behrends and Bl\"omer \cite{Alt-HausdorffDistance-1995}, the critical points for directed Hausdorff distance $\d_{\overrightarrow{H}}(A,B)$ of two pairwise disjoint sets of line segments $A$ and $B$ is either at some vertex of $A$ or at some intersection point of $A$ with the boundary of a Voronoi cell of $B$.
Thus, we need a predicate for encoding the first type of event where the distance is assumed at a vertex of $A$.
Additionally, we need a predicate for testing if a line supporting an edge intersects the intersection of two stadiums; see Figure~\ref{fig:triple_intersection} for an illustration in $\RR^2$.

\begin{figure}
    \centering
    \includegraphics{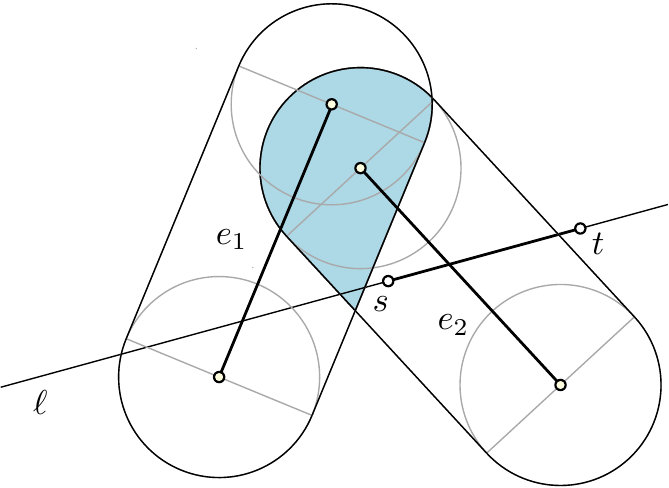}
    \caption{
    Illustration of the predicate $\predSLh$ in $\RR^2$: The predicate evaluates to true if and only if the triple intersection of the line $\ell$ supporting $\overline{st}$ with the two stadiums centered at edges $e_1$ and $e_2$ is non-empty. Note that $\overline{st}$ may lie outside of the intersection.}
    \label{fig:triple_intersection}
\end{figure}

Consider any two polygonal curves $s \in \XX^d_m$ and $q \in \XX^d_k$. In order to encode the intersection of polygonal curves with metric balls under the Hausdorff metric, we will first  define a subset of $\RR^d$, a \emph{double-stadium}, defined by two line segments $\{e_1, e_2\}$ and a radius $r$ as
\[
 D_{r,2}(e_1, e_2) = \stadium_r(e_1) \cap \stadium_r(e_2).
\]
We use the notation $\overline{st} \in D_{r,2}(e_1,e_2)$ to indicate that the line $\ell(\overline{st})$ which extends $\overline{st}$ intersects with the double stadium, i.e.\ it fulfills $\ell(\overline{st}) \cap D_{r,2}(e_1,e_2) \neq \emptyset$.

We will make use of the following predicates:
\begin{enumerate}
\item[$\predVEh$] \emph{(Vertex-edge (horizontal))} Given an edge of $s$, $\overline{s_j
s_{j+1}}$, and a vertex $q_i$ of $q$, this predicate returns true iff there
exist a point $p \in \overline{s_j s_{j+1}}$, such that $\|p-q_i\| \leq r$.
\item[$\predVEv$] \emph{(Vertex-edge (vertical))} Given an edge of $q$, $\overline{q_i
q_{i+1}}$, and a vertex $s_j$ of $s$, this predicate returns true iff there
exist a point $p \in \overline{q_i q_{i+1}}$, such that $\|p-s_j\| \leq r$.
\item[$\predSLh$] \emph{($d$-stadium-line (horizontal))} given an edge of $q$,  $\overline{q_i q_{i+1}}$, and two edges of $s$, $\{e_1, e_2\} \subset E(s)$, this predicate is equal to $\overline{q_i q_{i+1}} \in D_{r,2}(e_1, e_2)$.   
\item[$\predSLv$] \emph{($d$-stadium-line (vertical))} given one edge of $s$,  $\overline{s_j s_{j+1}}$, and two edges of $q$, $\{e_1, e_2\} \subset E(q)$, this predicate is equal to $\overline{s_j s_{j+1}} \in D_{r,2}(e_1,e_2)$.   
\end{enumerate}

\begin{lemma}
\label{lemma:HausPrCorr}
For any two polygonal curves $s$, $q$, 
given the truth values of the predicates $\predVEh$, $\predSLh$ one can determine whether $\d_{\overrightarrow{H}}(q,s)\leq r$. Similarly, given the truth values of the predicates $\predVEv$, $\predSLv$ one can determine whether $\d_{\overrightarrow{H}}(s,q)\leq r$.
\end{lemma}
\begin{proof}
We first assume for the sake of simplicity that $q$ is a line segment with endpoints $q_1$ and $q_2$.
We claim that $\d_{\overrightarrow{H}}(q,s)\leq r$ if and only if there exists a sequence of edges $\overline{s_{j_1}s_{({j_1}+1)}},$ $\overline{s_{j_2}s_{(j_{2}+1)}},\ldots ,$ $\overline{s_{j_v}s_{(j_{v}+1)}}$ for some integer value $v$,
such that the predicates $\predVEh(q_1,\overline{s_{j_1}s_{(j_{1}+1)}}) $, $\predVEh(q_2,\overline{s_{j_v}s_{(j_{v}+1)}}) $
both evaluate to true and the conjugate 
\[
    \bigwedge_{i=1}^{v-1} \predSLh(\overline{q_1,q_2},\overline{s_{j_i} s_{{(j_{i}+1)}}},\overline{s_{j_{i+1}} s_{(j_{i+1}+1)}})
\]
evaluates to true. 
Assume such a sequence of edges exists.  In this case, there exists a sequence of points $p_1,\dots,p_{v}$ on the line supporting $q$, with $p_1=q_1$, $p_{v}=q_2$ and such that for $1 \leq i< v$:
$p_i, p_{i+1} \in \stadium_r(\overline{s_{j_i} s_{j_{i+1}}})$. That is, two consecutive points of the sequence are contained in the same stadium.
Indeed, for $i=1$ we have $p_1=q_1$ and $q_1, p_2 \in \overline{s_{j_1}s_{(j_{1}+1)}}$ by the corresponding $\predVEh$ and $\predSLh$ predicates:
\[
\predVEh(q_1,\overline{s_{j_1}s_{(j_{1}+1)}}), \quad \predSLh(\overline{q_1,q_2},\overline{s_{j_1} s_{{(j_{1}+1)}}},\overline{s_{j_{2}} s_{(j_{2}+1)}}).
\]

Likewise, for $i=v-1$, it is implied by the corresponding predicates $\predVEh(q_2,\overline{s_{j_v}s_{(j_{v}+1)}}) $ and $\predSLh(\overline{q_1,q_2},\overline{s_{j_{v-1}} s_{{(j_{v-1}+1)}}},\overline{s_{j_{v}} s_{(j_{v}+1)}})$. For the remaining $1 < i < v-1$, it follows from the conditions given by the specified $\predSLh$ predicates.
Now, since each stadium is a convex set, it follows that each line segment connecting two consecutive points of this sequence $p_i$, $p_{i+1}$ is contained in one of the stadiums. Note that the set of line segments obtained this way forms a connected polygonal curve which fully covers the line segment $q$.  It follows that 
\[ q \subseteq \bigcup_{0 \geq i < v } \overline{p_i p_{i+1}} \subseteq \bigcup_{0 \geq i < v }  \stadium_r(\overline{s_{j_i} s_{j_{i+1}}}) \subseteq \bigcup_{e \in E(s)} \stadium_r(e). \]
Therefore, any point on $q$ is within distance $r$ of some point on $s$ and thus $\d_{\overrightarrow{H}}(q,s)\leq r$.

Now, for the other direction of the proof, assume that  $\d_{\overrightarrow{H}}(q,s)\leq r$. The definition of the directed Hausdorff distance implies that 
\[ 
q \in \bigcup_{e \in E(s)} \stadium_r(e), 
\]
since any point on the line segment $q$ must be within distance $r$ to some point on the curve $s$.
Consider the intersections of the line segment $q$ with the boundaries of stadiums 
\[ 
q \cap \bigcup_{e \in E(s)} \partial\left( \stadium_r(e)\right).
\] 
Let $w$ be the number of intersection points and let $v=w+2$.
We claim that this implies that there exists a sequence of edges $\overline{s_{j_1}s_{({j_1}+1)}},$ $\overline{s_{j_2}s_{(j_{2}+1)}},\ldots ,$ $\overline{s_{j_v}s_{(j_{v}+1)}}$ with the properties stated above. Let $p_1=q_1$ and let $p_v=q_2$ and let $p_i$ for $1<i<v$ be the intersection points ordered in the direction of the line segment $q$. By construction, it must be that each $p_i$ for $1<i<v$ is contained in the intersection of two stadiums, since it is the intersection with the boundary of a stadium and the entire edge is covered by the union of stadiums. Moreover, two consecutive points $p_i$, $p_{i+1}$ are contained in exactly the same subset of stadiums---otherwise there would be another intersection point with the boundary of a stadium in between $p_i$ and $p_{i+1}$. This implies a set of true predicates of type $\predSLh$ with the properties defined above. The predicates of type $\predVEh$ follow trivially from the definition of the directed Hausdorff distance. This concludes the proof of the other direction.


In general, for any polygonal curve $q \in \XX^d_k$ with vertices $q_1,\ldots,q_k$, we have that 
\begin{equation*}
\d_{\overrightarrow{H}}(q,s)\leq r \iff \bigwedge_{i=1}^{k-1}\left[ \d_{\overrightarrow{H}}(\overline{q_iq_{i+1}},s)\leq r \right].
\end{equation*}
Thus, we can apply the arguments above to each edge of $q$ individually. Similarly, we can prove that given the truth values of the predicates $\predVEv$, $\predSLv$ one can determine whether $\d_{\overrightarrow{H}}(s,q)\leq r$, by an argument symmetric to the above.
\end{proof}

\subsection{Hausdorff distance VC dimension bound}

Now, we want to show that we can compute a representation of the interval of intersection of a line and a capped cylinder using only $O(d)$ simple operations. This representation then allows us to compare such intervals using Lemma \ref{lem:comparison}.
The appropriate ground set is over two points $q_j, q_t \in \RR^d$, where for notational simplicity we reuse $\XX_2^d$. Furthermore, for each segment $\overline{st} \in \XX_2^d$, recall that $\ell(\overline{st})$ is the line that supports it.

\begin{lemma}\label{lem:cylinder_simple}
	Given a line $\ell(\overline{st})$ with $\overline{st} \in \XX_2^d$ and a capped cylinder $R_r(\overline{uv})$ with $\overline{uv} \in \XX_2^d$, the intersection $\ell(\overline{st}) \cap R_r(\overline{uv})$ of those two objects is either
\[
	 \left\{s + (t-s)x \; \middle| \; x \in [a + \sqrt{b}, c + \sqrt{d}] \subseteq \RR\right\},
\]
where $a, b, c, d \in \RR$ can be computed using $O(d)$ simple operations, or it is empty.
\end{lemma}
\begin{proof}
	We first compute the intersection of the infinite cylinder $C_r(\overline{uv})$ with the line $\ell(\overline{st})$. Let $f(x) = u - (v - u) \cdot x$ be the line $\ell(\overline{uv})$ parametrized by $x \in \RR$ and $g(y) = s + (t-s) \cdot y$ the line $\ell(\overline{st})$ parametrized by $y \in \RR$. We describe all values $x,y$ parameterizing points in this intersection by quantifying the boundaries of this set.  All points in the intersection of $\ell(\overline{st})$ with the boundary of the infinite cylinder $C_r(\overline{uv})$ are described by
\begin{align*}
	&\quad \|g(y) - f(x) \|^2 = r^2 \\
\Longleftrightarrow &\quad \|s + (t-s) \cdot y - u - (v - u) \cdot x \|^2 = r^2 \\
\Longleftrightarrow &\quad \sum_{i=1}^d \big(s_i + (t_i - s_i) \cdot y - u_i - (v_i - u_i) \cdot x \big)^2 = r^2.
\end{align*}
Let $z_i(y) = s_i + (t_i - s_i) \cdot y - u_i$, we obtain
\begin{align*}
&\quad \sum_{i=1}^d \big(z_i(y) - (v_i - u_i) \cdot x \big)^2 = r^2 \\
\Longleftrightarrow &\quad \left(\sum_{i=1}^d (v_i-u_i)^2 \cdot x^2 - 2 z_i(y) (v_i-u_i) \cdot x + z_i(y)^2\right) - r^2 = 0.
\end{align*}
For any fixed $y$, this is a quadratic equation in $x$ and the discriminant is
\[
	h(y) = \left(\sum_{i=1}^d 2 z_i(y) (v_i-u_i)\right)^2 - 4 \left(\sum_{i=1}^d (v_i-u_i)^2\right) \left(\sum_{i=1}^d z_i(y)^2\right) - r^2.
\]
Note that the quadratic equation has one solution exactly for those points on $\ell(\overline{st})$ which have distance $r$ from $\ell(\overline{uv})$, because the ball around those points intersects $\ell(\overline{uv})$ exactly once. Those are also the points which define the boundary of $\ell(\overline{st}) \cap R_r(\overline{uv})$.
Thus, we want to solve $h(y) = 0$. As $z_i(y)$ is linear in $y$, we obtain a quadratic equation in $y$. Note that all coefficients of the quadratic equation can be computed in $O(d)$ simple operations.
Both solutions of this equation are of the form $\alpha \pm \sqrt{\beta}$. If $\beta < 0$ then the intersection is empty. Otherwise, we obtain an intersection interval $[\alpha - \sqrt{\beta}, \alpha + \sqrt{\beta}]$ for the infinite cylinder.

To obtain the intersection with the original cylinder, we first compute the intersection of $\ell(\overline{st})$ with the top and bottom hyperplane of the cylinder. The two planes are given by all $p \in \RR^d$ which satisfy $(p-u)\cdot(v-u) = 0$ and $(p-v)\cdot(v-u) = 0$, respectively. By plugging the line equation into the hyperplane formulas, we get the intersection points. For the first plane we thereby obtain
\[
(g(y) - u) \cdot (v-u) = 0 \quad \Longleftrightarrow \quad (s + (t-s) \cdot y - u) \cdot (v-u) = 0 \quad \Longleftrightarrow \quad y = -\frac{(s-u)\cdot(v-u)}{(t-s)\cdot(v-u)}.
\]
The intersection with the second plane is analogous. Both intersections can be computed with $O(d)$ simple operations. We now compare the intersection interval of the planes and of the infinite cylinder using Lemma \ref{lem:comparison} to obtain the final intersection interval.
\end{proof}

Additionally, the following lemma holds, which says that we can express an intersection of a ball and a line with an interval of the form as in the previous lemma.
\begin{lemma}\label{lem:ball_simple}
Given a line $\ell(\overline{st})$ with $\overline{st} \in \XX_2^d$ and a ball $B_r(c)$ centered at $c$, the intersection $\ell(\overline{st}) \cap B_r(c)$ of those two objects is either
\[
	 \left\{s + (t-s)x \; \middle| \; x \in [a + \sqrt{b}, c + \sqrt{d}] \subseteq \RR\right\},
\]
where $a, b, c, d \in \RR$ can be computed using $O(d)$ simple operations, or it is empty.
\end{lemma}
\begin{proof}
The intersection is given by the $x$ fulfilling $\|s + (t-s) \cdot x - c\|^2 \leq r^2$.  To determine the extremal values for $x$ which satisfy this inequality is a quadratic equation in $x$.  Solving it, we obtain an intersection interval as required.
\end{proof}

Having proven those technical lemmas, we are now ready to start our argument for bounding the VC dimension. We argue that the truth values for predicate $\predVEh$ over all possible inputs are uniquely defined by the  set $\predVEh^r(q,s)$. Similarly, the truth values for predicate $\predVEv$ are uniquely defined by the set $\predVEv^r(q,s)$. 

Then the predicates $\predSLh$ and $\predSLv$ induce sets (where effectively $\predSLv(q,s) = \predSLh(s,q)$):
\begin{itemize}
\item $\predSLh^r(q,s) = \{ (e_1, e_2, e_3) \in E(s) \times E(s) \times E(q) \mid  e_3 \in D_{r,2}(e_1,e_2)\}$.
\item $\predSLv^r(q,s) = \{ (e_1, e_2, e_3) \in E(q) \times E(q) \times E(s) \mid  e_3 \in D_{r,2}(e_1,e_2)\}$.
\end{itemize}

We require a technical proof, bounding the VC dimension of the range space defined on segments with ranges defined by double-stadiums.  
To this end, let 
\[
\mathcal{D}^d_{2} = \Big\{\{\overline{st} \in \XX_2^d \mid \ell(\overline{st}) \in D_{r,2}(e_1, e_2)\} \;\Big|\; e_1, e_2 \in \XX_2^d, \; r > 0\Big\}
\] 
be the families of subsets of line segments $\overline{st} \in \XX_2^d$ whose supported lines $\ell(\overline{st})$ intersect a common double stadium $D_{r,2}(e_1,e_2)$.  We are now ready to state and prove the following lemma.


\begin{lemma} \label{lem:VC-SSl}
The VC dimension of the range space $(\XX^d_2, \mathcal{D}^d_{2})$ and of the associated dual range space is $O(d^2)$.  
\end{lemma}

\begin{proof}
The predicate which determines whether a line $\ell$ intersects a double stadium $D_{r,2}(e_1,e_2)$ can be implemented by taking the logical-or over $O(1)$ calls to the following predicates (see Figure~\ref{fig:triple_intersection-sets} for an illustration):
\begin{itemize}
    \item $P_{BB}:$ checks whether $\ell$ intersects $D_{r,2}(e_1,e_2)$ in the intersection of two radius $r$ balls,
    \item $P_{RR}:$ checks whether $\ell$ intersects $D_{r,2}(e_1,e_2)$ in the intersection of two radius $r$ cylinders,
    \item $P_{RB}:$ checks whether $\ell$ intersects $D_{r,2}(e_1,e_2)$ in the intersection of one ball and one cylinder, both of radius $r$. 
\end{itemize}


\begin{figure}
    \centering
    \includegraphics[page=2,width=0.3\textwidth]{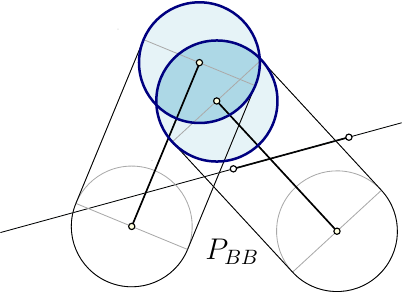}\hfill
    \includegraphics[page=3,width=0.3\textwidth]{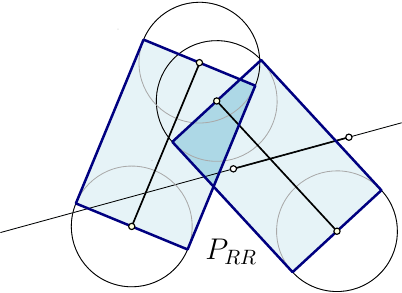}\hfill
    \includegraphics[page=4,width=0.3\textwidth]{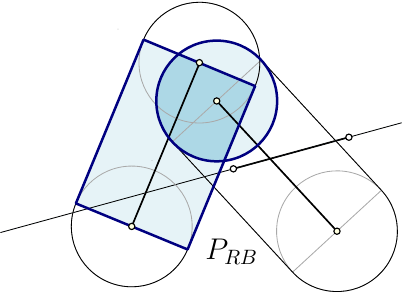}
    \caption{Illustration in $\RR^2$ of predicates used in the proof of Lemma~\ref{lem:VC-SSl} for the example given in Figure~\ref{fig:triple_intersection}.}
    \label{fig:triple_intersection-sets}
\end{figure}

For all predicates we first compute the intersection interval of the cylinder or ball using Lemma \ref{lem:cylinder_simple} or Lemma \ref{lem:ball_simple}. Applying Lemma \ref{lem:comparison}, we can then compute the intersection of these two intersection intervals by comparing their bounds, obtaining an interval of the form $[a + \sqrt{b}, c + \sqrt{d}]$. Again using Lemma \ref{lem:comparison}, we test if $a + \sqrt{b} \leq c + \sqrt{d}$, thereby checking if the intersection is non-empty. Thus, all three of the above predicates can be computed in $O(d)$ simple operations. Because each predicate requires $O(d)$ simple operations, and we need to perform a logical-or over $O(1)$ of these predicates, it implies range inclusion $e \in B_{r,2}$ and can be determined with $O(d)$ simple operations.  Hence by Theorem \ref{theorem:predicatevc} the VC dimension is $O(d^2)$.  Since an element of the dual range space is also defined by $O(d)$ real values, and the same operations can be applied, the dual range space also has VC dimension $O(d^2)$.  
\end{proof}

Using the above lemmas, we now get the following theorems.

\begin{theorem}
\label{Thaus1}
Let $\overrightarrow\RRR_{H,k}$ be the set of all \nnms, under the directed Hausdorff distance {\em from} polygonal curves in $\XX^d_k$. The VC dimension of $(\XX^d_m,\overrightarrow\RRR_{H,k})$ is $O(d^2 k^2 \log (dkm))$. 
\end{theorem}

\begin{proof} 
Let $S\subset \XX^d_m$ be a set of $t$ polygonal curves and let $q\in \XX^d_k$. By Lemma \ref{lemma:HausPrCorr}, the set $\{s\in S \mid \d_{\overrightarrow{H}}(q,s)\leq r\}$ 
is uniquely defined by the sets:
\[
\bigcup_{s\in S}\predVEh^r(q,s),\bigcup_{s \in S}P^r_7(q,s).
\]
By Lemma~\ref{lem:ball_simple}, the number of all possible sets $\bigcup_{r\geq0} \bigcup_{s\in S} P^r_3(q,s)$ is bounded by $\left(tm\right)^{O(d^2 k)}$. 
Furthermore, by Lemma \ref{lem:VC-SSl}, we are able to bound the number of all possible sets $ \bigcup_{r\geq0}
\bigcup_{s\in S} P^r_7(q,s)$ as $\left(tm\right)^{O(d^2 k^2)}$.  The $k^2$ term arises because we consider $\Theta(k^2)$ pairs $s_j,s_t$ for predicate $\predSLh$.  
Hence,
\[
	2^t\leq O((tm)^{d^2 k^2}) \implies t = {O}\left(d^2 k^2 \log(dk m)\right). \qedhere
\]
\end{proof}

\begin{theorem}
\label{Thaus2}
Let $\overleftarrow\RRR_{H,k}$ be the set of all \nnms, under the directed Hausdorff distance {\em to} polygonal curves in $\XX^d_k$. The VC dimension of $(\XX^d_m,\overleftarrow\RRR_{H,k})$ is $O(d^2 k \log (dkm))$. 
\end{theorem}

\begin{proof} 
Let $S\subset \XX^d_m$ be a set of $t$ polygonal curves and let $q\in \XX^d_k$. By Lemma \ref{lemma:HausPrCorr}, the set $\{s\in S \mid \d_{\overrightarrow{H}}(q,s)\leq r\}$ 
is uniquely defined by the sets:
\[
\bigcup_{s\in S}\predVEv^r(q,s),\bigcup_{s \in S}P^r_8(q,s).
\]
By Lemma~\ref{lem:ball_simple}, the number of all possible sets $\bigcup_{r\geq0} \bigcup_{s\in S} P^r_4(q,s)$ is bounded by $\left(tm\right)^{O(d^2 k)}$. 
Furthermore, by Lemma \ref{lem:VC-SSl}, we are able to bound the number of all possible sets $ \bigcup_{r\geq0}
\bigcup_{s\in S} \predSLv^r(q,s)$ as $\left(tm\right)^{O(d^2 k)}$.  This is only linear in $k$ since 
we only need to consider each of $k$ segments $s_j$ for predicates $\predSLv$.  
Now, 
\[
	2^t\leq O((tm)^{d^2 k}) \implies t = {O}\left(d^2 k \log(dkm)\right). \qedhere
\]
\end{proof}

\begin{theorem}
\label{Thaus3}
Let $\RRR_{H,k}$ be the set of all \nnms, under the symmetric Hausdorff distance in $\XX^d_k$. The VC dimension of $(\XX^d_m,\RRR_{H,k})$ is $O(d^2 k^2 \log (dkm))$. 
\end{theorem}
\begin{proof}
Lemmas \ref{lemma:HausPrCorr} imply that 
the set $\{s\in S \mid \d_{{H}}(q,s)\leq r\}$ 
is uniquely defined by the sets:
\[
\bigcup_{s\in S}\predVEh^r(q,s),\bigcup_{s\in S}\predVEv^r(q,s),\bigcup_{s\in S}\predSLh^r(q,s),\bigcup_{s\in S}\predSLv^r(q,s).
\]
Now bounding the number of all possible such sets, as we did in the proofs of Theorems \ref{Thaus1} and \ref{Thaus2}, implies the statement.
\end{proof}

\section{The Fr\'echet distance}
\label{predicates-frechet}
We consider the range spaces $(\XX^d_m, \RRR_{F_k})$ and $(\XX^d_m, \RRR_{wF_k})$, where $\RRR_{F_k}$ (resp. $\RRR_{wF_k}$) denotes the set of all \nnms, centered at curves in $\XX^d_k$, under the Fr\'{e}chet distance (resp. weak Fr\'{e}chet) distance.  

\subsection{Fr\'{e}chet distance predicates}
It is known that the Fr\'echet distance between two polygonal curves can be attained, either at a distance between their endpoints, at a distance between a vertex and a line supporting an edge, or at the common distance of two vertices with a line supporting an edge. The third type of event is sometimes called monotonicity event, since it happens when the Weak Fr\'echet distance is smaller than the Fr\'echet distance.
In this sense, our representation of the \nnm  of radius $r$ under the Fr\'echet distance is based on the following predicates, some of which we already used in the last section.
Let $s\in \XX_m^d$ with vertices $s_1,\ldots,s_{m}$ and $q\in \XX_k^d$ with vertices $q_1,\ldots,q_{k}$.

\begin{enumerate} 
\item[$\predVEh$] \emph{(Vertex-edge (horizontal))} As defined in Section \ref{section:haus}.

\item[$\predVEv$] \emph{(Vertex-edge (vertical))} As defined in Section \ref{section:haus}.

\item[$\predEs$] \emph{(Endpoints (start))} This
predicate returns true if and only if $\|s_1-q_1\| \leq r$. \label{ep}

\item[$\predEe$] \emph{(Endpoints (end))} This predicate returns true if and only if
$\|s_{m}-q_{k}\| \leq r$. \label{ep2}

\item[$\predMh$] \emph{(Monotonicity (horizontal))} Given two vertices of $s$, $s_j$ and
$s_t$ with $j<t$ and an edge of $q$, $\overline{q_i q_{i+1}}$, this predicate
returns true if there exist two points $p_1$ and $p_2$ on the {\em line\/} supporting
the directed edge, such that $p_1$ appears before $p_2$ on this line, and such
that $\|p_1 - s_j\| \leq r$ and $\|p_2-s_t\| \leq r$.  \label{hmp}

\item[$\predMv$] \emph{(Monotonicity (vertical))} Given two vertices of $q$, $q_i$ and
$q_t$ with $i<t$ and a directed edge of $s$, $\overline{s_j s_{j+1}}$, this
predicate returns true if there exist two points $p_1$ and $p_2$ on the {\em line\/}
supporting the directed edge, such that $p_1$ appears before $p_2$ on this
line, and such that $\|p_1 - q_i\| \leq r$ and $\|p_2-q_t\| \leq r$.
\label{vmp} \end{enumerate}

\begin{lemma}[Lemma 9, \cite{arxAD17}]
\label{lemma:lem9}
Given the truth values of all predicates
$\predVEh,\predVEv, \predEs, \predEe, \predMh, \predMv$ of two curves $s$ and $q$ for a fixed value of $r$, one can
determine if $d_{F}(s,q) \leq r$.  \end{lemma}

Predicates $\predVEh,\predVEv, \predEs, \predEe$ are sufficient for representing metric balls under the weak Fr\'echet distance. We include a proof for the sake of completeness.

\begin{lemma}
Given the truth values of all predicates $\predVEh,\predVEv, \predEs, \predEe$ of two curves $s$ and $q$ for a fixed value of $r$, one can determine if $\d_{wF}(s,q)\leq r$.
\end{lemma}
\begin{proof}
Alt and Godau \cite{altgodau-95} describe an algorithm for computing the Weak Fr\'echet distance which can be used here. In particular, one can construct an edge-weighted grid graph on the cells (edge-edge pairs) of the parametric space of the two polygonal curves and subsequently compute a bottleneck-shortest path from the pair of first edges to the pair of last edges along the two curves. We can use edge weights in $\{0,1\}$ to encode if the corresponding vertex-edge pair has distance at most $r$, as given by the predicates $\predVEh$ and $\predVEv$. If and only if there exists a bottleneck shortest path of cost $0$ and the endpoint conditions are satisfied (as given by the predicates $\predEs$ and $\predEe$), the Weak Fr\'echet distance between $q$ and $s$ is at most $r$.
\end{proof}

\subsection{Fr\'{e}chet distance  VC dimension bounds}\label{sec:weak:frechet}
\label{SSwF}
We first consider the range space ($\XX_m^d,\RRR_{wF,k})$, where $\RRR_{wF,k}$ is the set of all \nnms under the Weak Fr\'{e}chet distance centered at curves in $\XX_k^d$.  
The main task is to translate the predicates $\predVEh,\predVEv, \predEs, \predEe$ into simple range spaces, and then bound their associated VC dimensions.  
Consider any two polygonal curves $s \in \XX_m^d$ and $q \in \XX_k^d$. In order to encode the intersection of polygonal curves with metric balls, we will make use of the following sets:
\begin{itemize}
\item $\predVEh^r(q,s)=\left \{\stadium_r(\overline{s_i s_{i+1}})\cap V(q) \mid \overline{s_i s_{i+1}}\in E(s)     \right\} $,
\item $\predVEv^r(q,s)=\left \{\stadium_r(\overline{q_i q_{i+1}})\cap V(s) \mid  \overline{q_i q_{i+1}}\in E(q)     \right\} $.
\item  $\predEs^r(q,s)=\ball_r(q_1)\cap V(s) $,
\item $\predEe^r(q,s)= \ball_r(q_k)\cap V(s) $,
\end{itemize}



\begin{theorem}\label{ThmwF}
Let $\RRR_{wF,k}$ be the set of  
\nnms under the Weak Fr\'echet metric centered at polygonal curves in $\XX_k^d$. The VC dimension of  
$(\XX_m^d,\RRR_{wF,k})$ is ${O}\left(d^2 k \log (dkm)\right)$.
\end{theorem}
\begin{proof}
If $S$ is a set of $t$ polygonal curves of complexity $m$, the set $\{s\in S\mid   \d_{wF}(s,q)\leq r\}$ is uniquely defined by 
the sets 
\[\bigcup_{s\in S} \predVEh^r(q,s),\bigcup_{s\in S} \predVEv^r(q,s),
\bigcup_{s\in S} \predEs^r(q,s),\bigcup_{s\in S} \predEe^r(q,s) .\]

The number of all possible sets $\bigcup_{r\geq0}\bigcup_{s\in S} \predVEh^r(q,s)$ and 
the number of all possible sets $\bigcup_{r\geq0}\bigcup_{s\in S} \predVEv^r(q,s)$ are both bounded by $(tm)^{O(d^2k)} $ by 
Corollary \ref{cor:shapes-VC} using set $\stadium_r(\overline{st})$, 
and by considering the dual range space, respectively.

Notice that the number of all possible sets $\bigcup_{r\geq0}\bigcup_{s\in S} \predEs^r(q,s)$ is bounded by $(tm)^{O(d)}$. The same holds for the number of all possible sets $\bigcup_{r\geq0}\bigcup_{s\in S} \predEe^r(q,s)$. 

Hence,  
\[2^t\leq (tm)^{O(d^2 k)} 
\implies t = {O}\left(d^2k \log(dk m)\right). \qedhere \]
\end{proof}


\label{SSF}
We now consider the range space $(\XX^d_m, \RRR_{F,k})$, where $\RRR_{F,k}$ denotes the set of all \nnms, centered at curves in $\XX^d_k$, under the Fr\'{e}chet distance. 
The approach is the same as with the Weak Fr\'{e}chet distance, except we also need to bound VC dimension of range spaces associated with predicates $\predMh$ and $\predMv$ to encode monotonicity.  
While there exists geometric set constructions that are used in the context of range searching~\cite{AD18,arxAD17} we can simply appeal to Theorem \ref{theorem:predicatevc}.   

We need to define a set to represent predicates $\predMh$ and $\predMv$.
To this end, we again use $\XX_2^d$ to represent the set of all segments in $\RR^d$.
Then the ranges $\mathcal{M}$ are defined by sets $M_r(\overline{s t}) \in \mathcal{M}$, defined with respect to radii $r\geq0$ and line segments $\overline{st}$.  Specifically, 
$M_r(\overline{st}) \subset \XX_2^d$ so any $\{q_1,q_2\} \in M_r(\overline{st})$ satisfies that 
\begin{itemize}
\item  $\exists p_1, p_2 \in \ell$ where $\overline{st}$ supports $\ell$ such that: 
\item $\|p_1-q_1\|\leq r$ and $\|p_2-q_2\| \leq r$;
and 
\item $p_1$ is less than $p_2$ along the line as $\langle p_1, t-s \rangle \leq \langle p_2, t-s\rangle$.  
\end{itemize}
The predicate $\predMh$ is satisfied if $s_j, s_t \in M_r(\overline{q_i q_{i+1}})$ and predicate $\predMv$ is satisfied if $q_i, q_t \in M_r(\overline{s_j s_{j+1}})$.  

\begin{figure}[t]
    \centering
    \includegraphics[page=1]{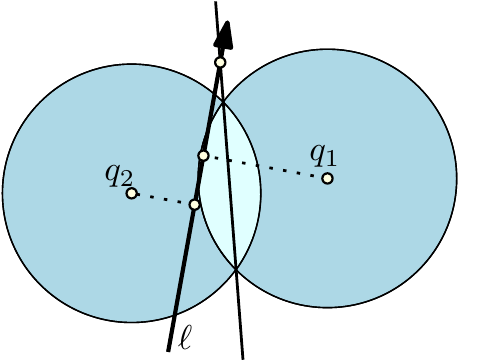}\quad\quad
    \includegraphics[page=2]{p5predicateproof}
    \caption{
    Illustration of predicate $\predMh$ in $\RR^2$ with line $\ell$ and the two disks centered at $q_1$ and $q_2$. In these examples, the projection of $q_2$ onto $\ell$ appears before the projection of $q_1$ onto $\ell$ along the direction of $\ell$ and the intersection of $\ell$ with the bisector lies outside of the lens formed by the two disks. On the left, the predicate is satisfied by setting $p_1=p_2=\pi_{\overline{st}}(q_1)$. On the right, the predicate evaluates to false.  
    }
    \label{fig:p5predicateproof}
\end{figure}

\begin{corollary}
\label{lem:M-VC}
The VC dimension of the range space $(\XX_2^d, \mathcal{M})$, and of the associated dual range space, is $O(d^2)$.
\end{corollary}
\begin{proof}

The corollary directly follows from Lemma \ref{lem:VC-SSl} by collapsing the stadiums to circles.

\end{proof}

We define sets to correspond with predicates $\predMh$ and $\predMv$:
\begin{itemize}
\item $\predMh^r(q,s) = \{\{s_j, s_t\} \in V(s) \times V(s) \mid (s_j,s_t) \in M_r(\overline{q_i q_{i+1}}) \text{ and } \overline{q_i q_{i+1}} \in E(q)\}$.
\item $\predMv^r(q,s) = \{\{q_i, q_t\} \in V(q) \times V(q) \mid (s_i,s_t) \in M_r(\overline{s_j s_{j+1}}) \text{ and } \overline{s_j s_{j+1}} \in E(s)\}$.
\end{itemize}

\begin{theorem}
\label{TFr}
Let $\RRR_{F,k}$ be the set of all \nnms, under the  
Fr\'{e}chet distance, centered at polygonal curves in $\XX^d_k$. The VC dimension of 
$(\XX_m^d,\RRR_{F,k})$ is ${O}\left(d^2 k^2 \log (d km)\right)$.
\end{theorem}

\begin{proof}
Due to Lemma~\ref{lemma:lem9}, if $S \subset \XX_m^d$ is a set of $t$ polygonal curves and $q\in \XX^d_k$, the set $\{s\in S\mid   \d_{F}(s,q)\leq r\}$ is uniquely defined by 
the sets 
\[
\bigcup_{s\in S} \predVEh^r(q,s),\bigcup_{s\in S} \predVEv^r(q,s) 
,
\bigcup_{s\in S} \predEs^r(q,s),\bigcup_{s\in S} \predEe^r(q,s),\bigcup_{s\in S} \predMh^r(q,s), 
\bigcup_{s \in S} \predMv^r(q,s)
.\]
As in the proof of Theorem~\ref{ThmwF}, the number of all possible sets \[\left(
\bigcup_{r\geq0} \bigcup_{s\in S} \predVEh(q,s), \bigcup_{r\geq0} \bigcup_{s\in S} \predVEv(q,s),
\bigcup_{r\geq0} \bigcup_{s\in S}
\predEs(q,s), \bigcup_{r\geq0} \bigcup_{s\in S} \predEe(q,s) \right)\] is bounded by $\left(tm\right)^{O(d^2 k)}$. 

By Corollary \ref{lem:M-VC} we are able to bound the number of all possible sets $ \bigcup_{r\geq0}
\bigcup_{s\in S} \predMh^r(q,s)$ as $\left(tm\right)^{O(d^2 k^2)}$.  The $k^2$ term arises because we consider $\Theta(k^2)$ pairs $s_j,s_t$ for predicate $\predMh$.  
And because this bound is proven using Theorem \ref{theorem:predicatevc}, then it applies to the dual range space, and we also bound the number of possible sets in $\bigcup_{r\geq0} \bigcup_{s\in S} \predMv^r(s,q)$ as also $\left(tm\right)^{O(d^2 k^2)}$.  
So ultimately, 
\[
2^t\leq (tm)^{O(d^2 k^2)} 
\implies t = {O}\left(d^2 k^2 \log(d k m)\right). \qedhere
\]
\end{proof}

\section{Lower bounds}

Our lower bounds are constructed in the simplified setting that either $k=1$ or $m=1$, i.e., either the ground set or the curves defining the metric ball consist of one vertex only. In this case, all of our considered distance measures (except for one direction of the directed Hausdorff distance) are equal:

\begin{lemma}\label{lem:allEqual}
Let $p \in \RR^d$, $q \in \XX^d_k$, let
$ r = \min_{s \in V(q)} \|s-p\| $. It holds that
\[ r = d_{dH}(q,p) = d_{\overrightarrow{H}}(q,p) = d_{dF}(q,p) = d_{F}(q,p) = d_{wF}(q,p) = d_{dH}(p,q).\]
\end{lemma}

\begin{proof}
In the discrete case we interpret $q \in \XX^d_k$ as an ordered or unordered sequence of points in $\RR^d$. In this case, the proof follows directly from Definitions  (Section~\ref{prelim}). In the continuous case we interpret $q \in \XX^d_k$ as a continuous polygonal curve. In this case, the proof follows directly from the definitions and from the convexity of the Euclidean ball of radius $r$ centered at the point $p$. If and only if all vertices of $q$ are contained in this ball, the distance is less or equal $r$.
\end{proof}

Because of the above lemma, any lower bound that we prove for the Hausdorff distance in the discrete setting automatically extends to the other distance measures.

\begin{lemma}\label{lem:lbHaus1}
Let $\RRR_{dH,k}$ be the set of all \nnms, under the Hausdorff distance, centered at point sets in $\XX^2_k$. 
The VC-dimension of the range space $(\XX^2_m,\RRR_{dH,k})$ is $\Omega(k)$.
\end{lemma}

\begin{proof}
We construct a set of $k$ points in $\RR^2$ that can be shattered by the ranges in $\RRR_{dH,k}$. The basic idea is that the ranges behave like convex polygons with $k$ facets. In particular, the set of points contained inside the range centered at a curve $q$, is equal to the intersection of a set of equal-size Euclidean balls centered at the vertices of $q$. For the construction we position a set $P$ of $k$ points on a unit circle, see Figure~\ref{fig:shattering_curve}. Let $R \gg 4$ be a parameter of the construction. For representing any subset of $P$ we construct $q$ using $k$ vertices (in any order) placed on the origin-centered circle of radius $R-1$. In particular, we can force any $p_0 \in P$ to be excluded from the metric ball under the Hausdorff distance and of radius $R-\eps$, for some $\eps>0$, by placing a vertex on the line through the origin that contains $p_0$ and by adding this vertex to the vertex set of $q$.  Using the $k$ vertices in $q$ we can specifically exclude any subset of up to $k$ points from $P$ by such a construction, and by placing a vertex of $q$ at the origin it will not exclude any points.  Hence any set $P$ on the unit circle of size $k$ can be shattered.  
\begin{figure}
    \centering
    \includegraphics{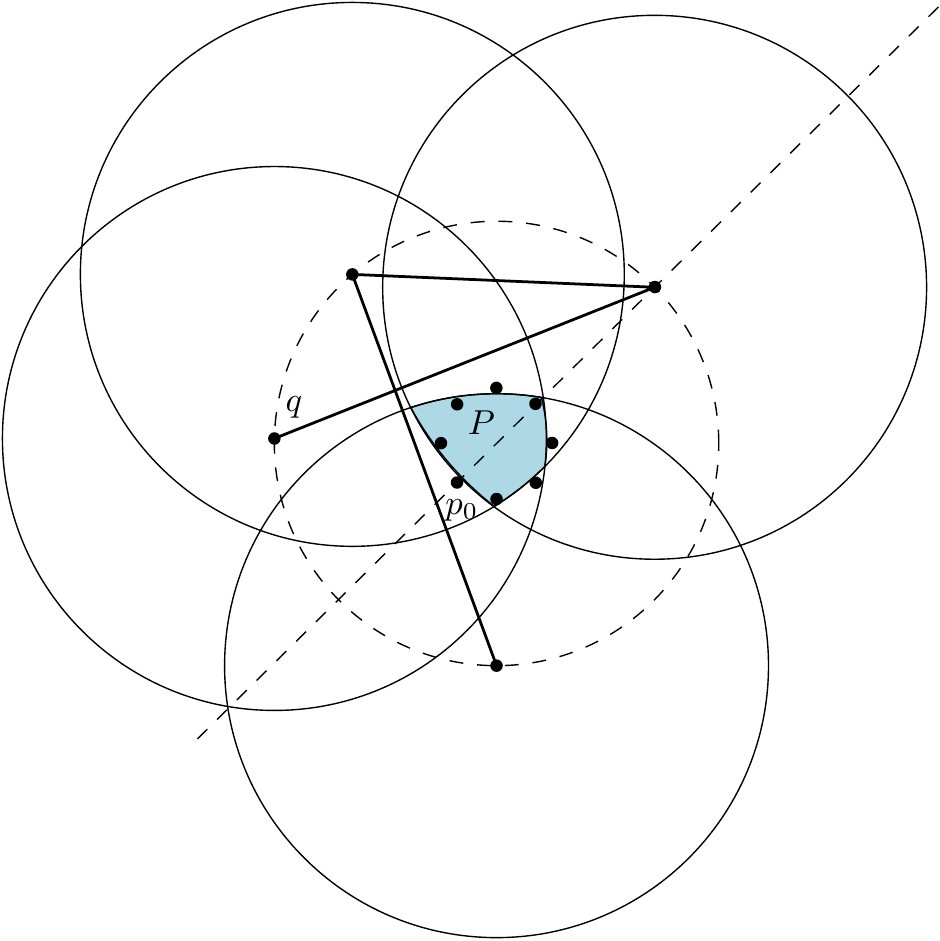}
    \caption{A curve $q$ with metric ball of radius $R-\eps$ containing a subset of $P$. The shaded area is the set of points that are contained inside the metric ball.}
    \label{fig:shattering_curve}
\end{figure}
\end{proof}

\begin{lemma}\label{lem:lbHaus2}
Let $\RRR_{dH,k}$ be the set of all $\nnms$, under the Hausdorff distance, centered at discrete point sets in $\XX^2_{k}$. The VC dimension of the range space $(\XX^2_m,\RRR_{dH,k})$ is $\Omega(\log m)$.
\end{lemma}

\begin{proof}
Theorem \ref{lem:lbHaus1} and \cite[Lemma 5.18]{H11}, which bounds the VC dimension of the dual range space as a function of the VC dimension of the primal space, imply the theorem.
\end{proof}

\begin{theorem}\label{ThmLbs}
The VC-dimension of the range spaces $(\XX^2_m,\RRR_{dF,k})$, $(\XX^2_m,\RRR_{dH,k})$, $(\XX^2_m,\RRR_{wF,k})$, $(\XX^2_m,\RRR_{F,k})$, and $(\XX^2_m,\RRR_{H,k})$ is $\Omega(\max(k,\log m))$.
\end{theorem}
\begin{proof}
The statement follows by applying Lemmas~\ref{lem:lbHaus1} and \ref{lem:lbHaus2} together with Lemma~\ref{lem:allEqual}.
\end{proof}

\begin{lemma}\label{lem:lbHaus3}
Let $\RRR_{dH,k}$ be the set of all \nnms, under the Hausdorff distance centered at point sets in $\XX^d_k$. 
For $d\geq 4$, the VC-dimension of the range space $(\XX^d_m,\RRR_{dH,k})$ is $\Omega(d k \log k)$.
\end{lemma}

\begin{proof}
As in the proof of Lemma~\ref{lem:lbHaus1}, our construction is set in the simplified setting where $m=1$, i.e., the ground set corresponds to points in $\RR^d$.
We now show the theorem by reducing to a recent lower bound of Csikos \etal~\cite{CKM18} which gave an $\Omega(d k \log k)$ lower bound for a related range space for $d\geq 4$.  This is defined on a ground set $P \subseteq \RR^d$ with ranges $\RRR_k$ defined so each range $R \in \RRR_k$ is the intersection of $k$ halfspaces.    
Recall that the construction in the proof of Lemma~\ref{lem:lbHaus1} used the fact that for $d=2$ the ranges behave like convex polygons. We can observe a similar behaviour in higher dimensions. In particular, Lemma~\ref{lem:allEqual} implies that any range in $\RRR_{dH,k}$ corresponds to the intersection of $k$ balls in $\RR^d$ (centered at vertices of $q$). Observe that for a sufficiently large fixed radius $R$, for any point set $P \subseteq \RR^d$ and for any halfspace $H$, we can find a ball of radius $R$ which has the same inclusion properties as $H$. 
Finally, the lower bound by Csikos \etal~\cite{CKM18} shows that there exist a set $P$ of $\kappa = \Omega(d k \log k)$ points  which can be shattered by such ranges.   
\end{proof}

\begin{lemma}\label{lem:lbHaus4}
Let $\RRR_{dH,k}$ be the set of all $\nnms$, under the Hausdorff distance, centered at point sets in $\XX^d_{k}$. For $d \geq 4$ The VC dimension of the range space $(\XX^d_m,\RRR_{dH,k})$ is $\Omega(\log dm)$.
\end{lemma}

\begin{proof}
Theorem \ref{lem:lbHaus3} and \cite[Lemma 5.18]{H11}, which bounds the VC dimension of the dual range space as a function of the VC dimension of the primal space, imply the theorem.
\end{proof}

\begin{theorem}\label{ThmLbsExt}
For $d\geq 4$, the VC-dimension of the range spaces $(\XX^d_m,\RRR_{dF,k})$, $(\XX^d_m,\RRR_{dH,k})$, $(\XX^d_m,\RRR_{wF,k})$, $(\XX^d_m,\RRR_{F,k})$, and $(\XX^d_m,\RRR_{H,k})$ is $\Omega(\max(d k \log k,\log dm))$.
\end{theorem}
\begin{proof}
The statement follows by applying Lemmas~\ref{lem:lbHaus3} and \ref{lem:lbHaus4} together with Lemma~\ref{lem:allEqual}. 
\end{proof}

\section{Implications}
\label{sec:implications}
In this section we demonstrate that bounds on the VC-dimension for the range space defined by metric balls on curves immediately implies various results about prediction and statistical generalization over the space of curves.   
In the following consider a range space $(X,\RRR)$ with a ground set $X$ of curves, where $\RRR$ are the ranges corresponding to metric balls for some distance measure we consider, and the VC-dimension is bounded by $\nu$.

This section discusses accuracy bounds that depend directly on the size $n = |X|$ and the VC-dimension $\nu$.  They will assume that $X$ is a random sample of some much larger set $X_{\textsf{big}}$ or an unknown continuous generating distribution $\mu$.  Under the randomness in this assumed sampling procedure, there is a probability of failure $\delta$ that often shows up in these bounds, but is minor since it shows up as $\log(1/\delta)$.  

These bounds take two closely-linked forms.  First, given a limited set $X$ from an unknown $\mu$, then how accurate is a query or a prediction made using only $X$.  Second, given the ability to draw samples (at a cost) from an unknown distribution $\mu$, how many are required so the prediction on the set of samples $X$ has bounded prediction error.  

Such large data sets of curves are now common place in many structured data applications.  For instance, the millions of ride-sharing trips taken every day, or the GPS traces Apple and Google and others collect on users' phones, or the tracking of migrating animals.  Because this data has a complex structure, and each associated curve may be large (i.e., $m$ is large), it is not clear how well analyses on families of such curves can provably generalize to predict new data.  The theme of the following results, as implied by our above VC-dimension results, is that if these families of curves are only inspected with or queried with curves with a small number of segments (i.e., $k$ is small), then the VC-dimension of the associated range space $\nu = O(k \log km)$ or $O(k^2 \log km)$ is small, and that such analyses generalize well.  We show this in several concrete examples.


\subparagraph*{Approximate range counting on curves.}
Given a large set of curves $X$ (of potentially very large complexity $m$), and a query curve $q$ (with smaller complexity $k$) we would like to approximate the number of curves nearby $q$.  For instance, we restrict $X$ to historical queries at a certain time of day, and query with the planned route $q$, and would like to know the chance of finding a carpool.  VC-dimension $\nu$ of the metric balls shows up directly in two analyses.  First, if we assume $X \sim \mu$ where $\mu$ is a much larger unknown distribution (but the real one), then we can estimate the accuracy of the fraction of all curves in this range within additive error $O(\sqrt{(1/|X|)(\nu + \log(1/\delta))})$.  
On the other hand, if $X$ is too large to conveniently query, we can sample a subset $S \subset X$ of size $O((1/\eps^2)(\nu + \log (1/\delta)))$ and know that the estimate for the fraction of curves from $S$ in that range is within additive $\eps$ error of the fraction from $X$.  Such sampling techniques have a long history in traditional databases~\cite{Olk93}, and have more recently become important when providing online estimates during a long query processing time as incrementally increasing size subsets are considered~\cite{BlinkDB}.  Ours provides the first formal analysis of these results for queries over curves.  

\subparagraph*{Density estimation of curves.}
A related task in generalization to new curves is density estimation.  Consider a large set of curves $X$ which represent a larger unknown distribution $\mu$ that models a distribution of curves; we want to understand how unusual a new curve $q$ would be, given we have not yet seen exactly the same curve before.  One option is to use the distance to the ($k$th) nearest neighbor curve in $X$, or a bit more robust option is to choose a radius $r$, and count how many curves are within that radius (e.g., the approximate range counting results above). 

Alternatively, for $X \subset \MM$, consider now a kernel density estimate $\kde_X : \MM \to \RR$ defined by $\kde_X(p) = \frac{1}{n} \sum_{p \in P} K(x,p)$ with kernel $K(x,p) = \exp(-\d(x,p)^2)$ (where $\d$ is some distance of choice among curves, e.g., $\d_F$).  The kernel is defined such that each superlevel set $K_x^\tau = \{p \in \MM \mid K(x,p) \geq \tau\}$ corresponds with some range $R \in \RRR$ so that $R \cap X = K_x^\tau \cap X$.  
Then a random sample $S \subset X$ of size $O((1/\eps^2) (\nu + \log \frac{1}{\delta}))$ satisfies that $\|\kde_X - \kde_s\|_\infty \leq \eps$~\cite{JKPV11}.  Thus, again the VC-dimension $\nu$ of the metric balls directly influences this estimates accuracy, and for query curves with small complexity $k$, the bound is quite reasonable.


\subparagraph*{Sample complexity for classification of curves.}
Now consider the problem of classifying curves representing trajectories of people or animals.  For instance with individuals who enable GPS on their cell phone they can label some work-to-home trajectories (as $\chi(x) =+1$) or as other trips ($\chi(x)=-1$).  Then on unlabeled trips we can potentially predict which are work-to-home trajectories to build traffic and commute time models without manually labeling all routes.  Similar tasks may be useful for normal ($\chi(x) = +1)$ versus nefarious ($\chi(x)=-1$) activities when tracking people in an airport or a hostile zone.
In each of these cases we may either have a very large number of labeled instances, and may want to sample them to some manageable size, or we may only have a limited number of samples, and want to know how much accuracy to trust based on the sample size.  All of these bounds are controlled by the VC-dimension of the family of classifiers used to make the prediction.  For trajectories, a sensible family of classifiers would be the ranges $\RRR$ defined by metric balls.  

That is consider some labeling function $\chi : X \to \{-1, +1\}$; now we say a range $R \subset \RRR$ misclassifies an object $x \in X$ if $x \in R$ and $\chi(x) = -1$ or $x \notin R$ and $\chi(x) = +1$.  
If there exists a range $R \subset \RRR$ such that all $x \in X \cap R$ have $\chi(x) = +1$ and all $x' \in X \setminus R$ have $\chi(x') = -1$; we say such a range \emph{perfectly separates} $(X,\chi)$. Then a random sample $Y \subset X$ of size $O((\nu/\eps) \log (\nu/\eps \delta))$~\cite{HW87} ensures that, with probability at least $1-\delta$, any range $R' \subset \RRR$ which perfectly separates $(Y, \chi)$ misclassifies at most $\eps n$ points in $X$.

Consider a random sample $Y \subset X$ of size $O((1/\eps^2)(\nu + \log \frac{1}{\delta}))$.  For any range $R \subset \RRR$, if the fraction of points in $Y$ is $|R \cap Y|/|Y| = \eta$, then with probability at least $1-\delta$, the fraction of points in $X$ is $|R \cap X|/|X| \in [\eta - \eps, \eta+\eps]$; that is its off by at most and $\eps$-fraction~\cite{LLS01,HS11}.  
If there is a labeling $\chi : X \to \{-1,+1\}$, this notably includes the range $R \in \RRR$ which misclassifies the least points (there may not be a perfect separator).  Hence a random sampling ensures at most an $\eps$-fraction more misclassified points are included in an estimate derived from this sample.  
Indeed the RBF kernel $K(x,p) = \exp(-\d(x,p)^2)$ defined above implies standard mechanism like kernel SVM or kernel perceptron~\cite{SS02} can be used to build classifiers, and together these bounds induce misclassification~\cite{LLS01} and margin approximation bounds~\cite{JKPV11}.  The small VC dimension $\nu$ implies they will generalize well.

\section{Acknowlegements}

We thank Peyman Afshani for useful discussions on the topic of this paper. We also thank the organizers of the 2016 NII Shonan Meeting ``Theory and Applications of Geometric Optimization'' where this research was initiated.

 \bibliography{biblio}
\bibliographystyle{plainurl}

\end{document}